\documentclass[5pt]{article}

\usepackage{amsthm,latexsym, amsfonts,amssymb,amstext, fullpage, hyperref, bbm} 
\usepackage{algorithmic, algorithm}
\usepackage{amsmath,amsfonts}
\usepackage[margin=1in]{geometry}
\usepackage{enumerate}

\usepackage{amsthm} 

\newcommand{\BIGOP}[1]{\mathop{\mathchoice%
{\raise-0.22em\hbox{\huge $#1$}}%
{\raise-0.05em\hbox{\Large $#1$}}{\hbox{\large $#1$}}{#1}}}

\newcommand{\bigtimes}{\BIGOP{\times}}

\newtheorem{theorem}{Theorem}[section]

\newtheorem{lemma}[theorem]{Lemma}

\newtheorem{claim}[theorem]{Claim}
\newtheorem{proposition}[theorem]{Proposition}
\newtheorem{definition}[theorem]{Definition}
\newtheorem{remark}[theorem]{Remark}

\newtheorem{fact}[theorem]{Fact}

\newcommand{\algwidth}{0.97\textwidth}

\newcommand{\E}{{\mathbb{E}}}
\newcommand{\eps}{\varepsilon}

\newcommand{\cB}{\mathcal{B}}

\newcommand{\cI}{\mathcal{I}}

\newcommand{\cG}{\mathcal{G}}
\newcommand{\cM}{\mathcal{M}}

\newcommand{\cX}{\mathcal{X}}

\newcommand{\cU}{\mathcal{U}}
\newcommand{\cT}{\mathcal{T}}

\newcommand{\cF}{\mathcal{F}}

\newcommand{\Ex}[1]{\mathop{\mathbb{E}}\displaylimits_{#1}}
\newcommand{\Div}[2]{\mathbb{D} \left (  #1 \| #2 \right) }

\newcommand{\Mbi}{M_{B_i}}
\newcommand{\mbi}{m_{B_i}}
\newcommand{\Mbiu}{M^u_{B_i}}

\newcommand{\mbiAll}{m^1_{B_i} m^2_{B_i},\ldots , m^{n_r}_{B_i}}
\newcommand{\MbiAll}{M^1_{B_i} M^2_{B_i},\ldots , M^{n_r}_{B_i}}
\newcommand{\Tsi}{T_{\sig(i)}}

\newcommand{\GiJ}{G(J_i)} 
\newcommand{\Gju}{G^u_{J_1}} 
\newcommand{\Gtu}{G^u_T} 
 
\newcommand{\Gjomu}{G^{-u}_{J_1}} 
\newcommand{\Gtmu}{G^{-u}_T} 

\newcommand{\sig}{\sigma}
\newcommand{\dr}{\Delta_r}
\newcommand{\tr}{5n_r\cdot \left(\sum_{k=0}^{r-1}\Delta_k^{1/2}\right) + 1}

\newcommand{\fs}{\ell\cdot n_r^2}
\newcommand{\psir}{\psi^i_r}

\usepackage{xspace}

\newcommand{ \Js}{J_1,\ldots, J_{n_r^4}}

\newcommand{ \J}{\mathbf{J}}
\newcommand{ \T}{\mathbf{T}}

\newcommand{\supp}{\mathsf{Supp}}

\newcommand{\papertitle}
{Welfare Maximization with Limited Interaction}

\begin{document}

\title{Welfare Maximization with Limited Interaction}

\author{
Noga Alon \thanks{ Tel Aviv University and Microsoft Research.
Research supported in part by BSF grant 2012/107, by ISF grant 620/13,
and by the Israeli I-Core program. nogaa@post.tau.ac.il. }
\and
Noam Nisan \thanks{Department of Computer Science, Hebrew University, and MSR Israel.
Supported by ISF grants 230/10 and 1435/14 administered by the Israeli Academy of Sciences.
noamn@microsoft.com.}
\and
Ran Raz \thanks{Weizmann Institute of Science, Israel, and the Institute for Advanced Study, Princeton, NJ. Research supported by the Israel Science Foundation grant No. 1402/14, by the I-CORE Program of the Planning and Budgeting Committee and the Israel Science Foundation, by the Simons Foundation, by the Fund for Math at IAS, and by the National Science Foundation grant No. CCF-1412958. ran.raz@weizmann.ac.il.}
\and
Omri Weinstein \thanks{Department of Computer Science, Princeton University. Research Supported by a Simons Fellowship award in TCS and a Siebel scholarship.
oweinste@cs.princeton.edu.} }


\maketitle

\begin{abstract}
We continue the study of welfare maximization in unit-demand (matching) markets, in a distributed information model 
where agent's valuations are unknown to the central planner, and therefore communication is required to determine an 
efficient allocation. Dobzinski, Nisan and Oren (STOC'14) showed that if the market size is $n$, then $r$ rounds of interaction 
(with logarithmic bandwidth) suffice to obtain an $n^{1/(r+1)}$-approximation to the optimal social welfare. In particular, 
this implies  that such markets converge to a stable state (constant approximation) in time logarithmic in the market size.

We obtain the first multi-round lower bound for this setup. We show that even if the allowable per-round bandwidth of each 
agent is $n^{\eps(r)}$, the approximation ratio of any $r$-round (randomized) protocol is no better than $\Omega(n^{1/5^{r+1}})$,
implying an $\Omega(\log \log n)$ lower bound on the rate of convergence of the market to equilibrium.

Our construction and technique may be of interest to round-communication tradeoffs in the more general setting of 
combinatorial auctions, for which the only known lower bound is for simultaneous ($r=1$) protocols \cite{DNO14}.
\end{abstract}

\thispagestyle{empty}
\titlepage

\section{Introduction}

This paper studies the tradeoff between the amount of communication and the number of rounds of interaction required to find an (approximately) optimal
matching in a bipartite graph.  In our model there are $n$ ``players'' and $m$ ``items''.  Each player initially knows a subset of the
items to which it may be matched (i.e. $ m$ bits of information).  The players communicate in rounds: in each round each player writes a message on a shared blackboard.  The message can only
depend on what the player knows at that stage: his initial input and all the messages by all other players that were written on the blackboard in previous rounds.

This problem was recently introduced by \cite{DNO14} as a simple market scenario: the players are unit-demand bidders 
and our goal is to find an (approximately) welfare-maximizing allocation of the items to players.
The classic auction of \cite{DGS86} -- that may be viewed as a simple Walrasian-like market process for this setting -- can be 
implemented as to find an approximately optimal allocation where each player needs only send $O(\log n)$ bits of communication (on the average).
The question considered by \cite{DNO14} was whether such a low communication burden suffices without using multiple rounds of interaction.
As a lower bound, they proved 
that a non-interactive protocol, i.e. one that uses a {\em single round of communication}, cannot get a $n^{1/2-\epsilon}$-factor
approximation (for any fixed $\epsilon>0$) with $n^{o(1)}$ bits of communication per player. As 
upper bounds they exhibited (I) an $O(\log n)$-round protocol, where each player sends $O(\log n)$ bits per round,
that gets a $\frac{1}{1-\delta}$-factor approximation (for any fixed $\delta>0$) and (II) for any fixed $r \ge 1$, an $r$-round protocol,
where each player sends $O(\log n)$ bits per round, that gets an $O(n^{1/(r+1)})$-approximation. 

The natural question at this point is whether there are $r$-round protocols with better approximation factors that still use $n^{o(1)}$
bits of communication per player.  This question was left open in \cite{DNO14}, where it was pointed out that
it was even open whether the exactly optimal matching
can be found by $2$-round protocols that use
$O(\log n)$ bits of communication per player. We answer this open problem by proving lower bounds for any fixed number of rounds. 

\vspace{0.15in}
\noindent
{\bf Theorem:} For every $r \ge 1$ there exists $\epsilon(r)=\exp(-r)$, such that every (deterministic or randomized)
$r$-round protocol requires $n^{\epsilon(r)}$ bits of communication per player in order to find a matching whose size is at least 
$n^{-\epsilon(r)}$ fraction of the optimal matching.
\vspace{0.15in}

\paragraph{Our Techniques.}
We construct a recursive family of hard distributions for every fixed number of communication rounds, and use information theoretic machinery
to analyze it. Our proof uses a type of direct-sum based round-reduction argument for multiparty communication complexity. 
Unlike standard round-elimination arguments in the two-party model, 
our instance  size (and thus the number of players) scales with the number of allowable rounds, and therefore eliminating a communication round essentially 
requires embedding a ``low dimensional" instance (with fewer players)  into a ``higher dimensional" protocol (operating over a larger
input), from its second round onwards.  In order to carry out such an embedding, we need a way of sampling the rest of the inputs to the
higher dimensional protocol (including the remaining players) \emph{conditioned on the first message of the protocol}, with no extra communication. 
The  main obstacle is that conditioning on the first message of the ``high dimensional" protocol  \emph{correlates  the private inputs of the players} (i.e., the inputs to 
the ``lower dimensional" protocol) \emph{with the ``missing" inputs}, and it is not hard to see that, in general, such sampling cannot be done without communication!
Circumventing this major obstacle calls for a subtle construction and analysis, which ensures the aforementioned correlations
remain ``local" and therefore allows to perform the embedding using a combination of private  and public randomness.
Our constructed family of distributions is designed to facilitate such embedding (using certain conditional independence properties) on one hand,
and yet retain a ``marginal indistinguishability" property which is essential
to keep the information argument above valid (we discuss this further in Section \ref{sec_main_lb}).

\subsection{More context and related models}

The bipartite matching problem is clearly a very basic one and obviously models a host of situations beyond the economic one
that was the direct motivation of \cite{DNO14} and of this paper.  Despite having been widely studied, even its algorithmic status is not well understood,
and it is not clear whether a nearly-linear time algorithm exists for it.  (The best 
known running time (for the dense case) is the 40-year old $O(n^{2.5})$ algorithm of \cite{HK73}, but for special cases
like regular or near-regular graphs nearly linear times are known (e.g. \cite{Alon03, Y13})).  In parallel computation, a major
open problem is whether bipartite matching  can be solved in deterministic parallel poly-logarithmic time
with a polynomial number of processors (Randomized parallel algorithms for the problem \cite{MVV87, KUW85} have been known for over 25 years).  
It was suggested in
\cite{DNO14} that studying the problem in the communication complexity model is an approach that might lead to algorithmic insights as well.

The bipartite matching problem has been studied in various other multi-party models that focus on communication as well. 
In particular,
strong and tight bounds for approximate matching are known in the 
weaker ``message passing'' or ``private channels'' models 
\cite{HBMQ13} that have implications to models of parallel and distributed computation .  Related work has also been done in networked distributed computing models, e.g., \cite{LPP08}.  ``One-way'' communication models are used to analyze streaming or
semi-streaming models and some upper bounds (e.g.,  \cite{K12})  as well as weak lower bounds \cite{GKK12} are known for approximate 
matchings in these models. For ``$r$-way" protocols, a super-linear communication lower bound was recently shown by \cite{GO13} for 
\emph{exact} matchings, in an incomparable model\footnote{Besides of the fact that the lower bound in \cite{GO13} applies only for testing \emph{exact} matchings 
and not approximate matchings, their model consists of a small number of parties (\emph{constant or logarithmic} in $n$) who are communicating in some fixed number 
of \emph{sequential} rounds (not simultaneous). 
The input itself of each player is therefore \emph{super-linear}
in the number of nodes of the input graph ($n$), and indeed they 
prove a super-linear communication lower bound, which is clearly 
impossible in our model. 
The \cite{GO13} model was motivated by streaming lower bounds and does not seem to capture the economic scenario we attempt to model in this paper (i.e., that of private-valuations) and therefore this result is incomparable to ours, as also evidenced by the distinct proof-techniques.}.  A somewhat more detailed survey of these related models can be found in the appendix of \cite{DNO14}.

It should be noted that the open problems mentioned above remain so even in the standard two-party setting where each of the two players holds all the information of $n/2$ of our players.  We do not know any better upper bounds than what is possible in the multi-player 
model, and certainly, as the model is stronger, no better lower bounds are known.  We also do not know whether our lower bound (or the single round one of \cite{DNO14}) applies also in this stronger two-player model.

\subsection{Open problems}

There are many open problems related to our work.  Let us mention a few of the most natural ones.
Our first open problem is closing the gap between our lower bound and the upper bound:  We show that $r=\Omega(\log\log n)$ rounds of communication 
are required to achieve constant approximation ratio using poly-logarithmic bits per player, while the upper bound is $r=O(\log n)$.  We believe 
that the upper bound is in fact tight, and improving the lower bound is left as our first and direct open problem.  

Another interesting direction is trying to extend our lower bound technique to obtain similar-in-spirit 
round-communication tradeoffs for the more general setup of combinatorial auctions, also studied by 
\cite{DNO14}. From a communication complexity perspective, lower bounds 
in this setup are more compelling, since player valuations require \emph{exponentially} many bits 
to encode, hence interaction has the potential to reduce the overall communication 
(required to obtain efficient allocations) from exponential to polynomial.  Indeed, it is shown in \cite{DNO14} that, in the 
case of \emph{sub-additive bidders}, there is an $r$-round randomized protocol that obtains an 
$\tilde{O}(r \cdot m^{1/(r+1)})$-approximation to the optimal social welfare, 
where in each round each player sends $poly(m, n)$ bits. 
Once again, an (exponential in $m$) lower bound on the communication was given only 
for the case of simultaneous protocols ($r=1$) and the natural question is to extend it to multiple rounds as well. 
 
A more general open problem advocated by \cite{DNO14} is to analyze the communication complexity of 
finding an {\emph exact} optimal matching.  One may naturally conjecture that $n^{\Omega(1)}$ rounds of interaction 
are required for this if each player only sends $n^{o(1)}$ bits in each round, but
no super-logarithmic bound is known.
The communication complexity of the problem without any limitation on the number of rounds is also open:
no significantly super linear, $\omega(n \log n)$, bound is known, while the best upper bound known is $\tilde{O}(n^{3/2})$.


\vspace{0.1in}\section{preliminaries}

We reserve capital letters for random variables, and calligraphic letters for sets. 
The $\ell_1$ (statistical) distance between two distributions in the same probability space is denoted
$|\mu - \nu| := \frac{1}{2}\cdot \sum_{a} |\mu(a)-\nu(a)|$. 
We write $X\perp Y \;| \; Z$ to denote that $X$ and $Y$ are statistically independent conditioned on the random variable $Z$.
For a vector random variable $X=X_1X_2\ldots X_s$, we sometimes use the shorthands $X_{\leq i}$ and $X_{-i}$ to denote  $X_1X_2\ldots X_i$ 
and $X_1X_2\ldots X_{i-1},X_{i+1},\ldots \ldots X_s$ respectively (similarly, $X^{-i} := X^1X^2\ldots X^{i-1}X^{i+1}\ldots X^{s}$). 
We write $A\in_R\cU$ to denote a uniformly distributed random variable over 
the set $\cU$. We use the terms ``bidders" and ``players" interchangeably
throughout the paper.

\vspace{0.1in}\subsection{Communication Model}

Our framework is the \emph{number-in-hand} (NIH) multiparty communication complexity 
model with shared blackboard. In this model, $n$ players receive inputs 
$(x_1,x_2,\ldots, x_n) \in \cX_1\times \cX_2\times \ldots \cX_n$ respectively. 
In our context, each of the $n$ players (bidders) is associated with a node $u\in U=[n]$ of some bipartite graph $G = (U,V,E)$,
and her input  is the set of incident edges on her node (her demand set of items in $V=[m]$). The players' goal is to compute a 
maximum set of disjoint connected pairs $(u,v) \in E(G)$, i.e., a maximum matching in $G$ (we define this formally below). 

The players communicate in some fixed number of rounds $r$, where in each communication round, players \emph{simultaneously} 
write (at most) $\ell$ bits each on a \emph{shared blackboard} which is viewable to all parties. We sometimes refer to the parameter $\ell$ as
the \emph{bandwidth} of the protocol.
In a deterministic protocol, each player's message should be completely determined by the content of the 
blackboard and her own private input $x_i$. 
In a randomized protocol, the message of each player may further depend on both public and private random coins.
When player's inputs are distributional ($(x_1,x_2,\ldots, x_n) \sim \mu$) which is the setting in this paper, we may assume 
without loss of generality that the protocol is \emph{deterministic}, since the averaging 
principle asserts that there is always some fixing of the randomness that will achieve the same performance with respect to $\mu$.
We remark that by the averaging principle, our main result applies to the randomized setting as well\footnote{More formally,
if there is a distribution $\mu$ on players inputs such that the approximation ratio of any $r$-round \emph{deterministic} protocol with respect to 
$\mu$ is at most $\alpha$ in expectation, then fixing the randomness of the protocol would yield a deterministic protocol with the same performance,
thus the former lower bound applies to randomized $r$-round protocols as well.}.

The transcript of a protocol $\pi$ (namely, the content of the blackboard) when executed on an input graph $G$ is denoted by $\Pi(G)$, or simply $\Pi$
when clear from context. 
At the end of the $r$'th communication round, a 
\emph{referee} (the ``central planner" in our context) computes a matching 
$\hat{\cM}(\Pi)$
, which is completely determined by $\Pi$. We call this the \emph{output} of the protocol. \\

We will be interested in protocols that compute \emph{approximate matchings}. To make this more formal,
let $\cG(n,m)$ denote the family of bipartite graphs on $(n,m)$-vertex sets respectively, and denote by
$\cF(n,m)$ the family of all matchings in $\cG(n,m)$ (not necessarily maximum matchings). Denote by  
$|\cM(G)|$ the size of a maximum matching in the input graph $G$. 
We require that the output of any protocol 
satisfies $\hat{\cM}(\Pi) \in \cF(n,m)$. The following definition is central to this work. 
\begin{definition}[Approximate Matchings]
We say that a protocol $\pi$ computes an $\alpha$-\emph{approximate matching} ($\alpha\geq 1$)
if $|\hat{\cM}(\Pi)\cap E(G)|$ is at least $\frac{1}{\alpha}\cdot |\cM(G)|$, i.e., if the number of matched pairs $(u,v) \in E(G)$ 
is at least a $(1/\alpha)$-fraction of the maximum matching in $G$. 
Similarly, when the input graph $G$ is distributed according to some distribution $\mu$ (i.e., $(x_1,x_2,\ldots , x_n)\sim \mu$),
we say that the \emph{approximation ratio} of $\pi$ is $\alpha$ 
if $$\Ex{G\sim \mu}[|\hat{\cM}(\Pi) \cap E(G)|] \geq \frac{1}{\alpha} \cdot \Ex{G\sim \mu}[|\cM(G)|].$$
\end{definition}
The \emph{expected matching size} of $\pi$ is $\E_\mu[|\hat{\cM}(\Pi) \cap E(G)|]$ (we remark that the ``hard" distribution we construct
in the next section will satisfy $|\cM(G)|\equiv n$ for all $G$ in the support of $\mu$, so the quantity $\Ex{G\sim \mu}[|\cM(G)|]$ will always be $n$). 
Note that these definitions in particular allow the protocol to be \emph{erroneous}, i.e., the referee is allowed to output ``illegal" pairs 
$(u,v)\notin E(G)$, but we only count the correctly matched pairs. Our lower bound holds even with respect to this more permissive model.



\vspace{0.1in}\subsection{Information theory}

Our proof relies on basic concepts from information theory. For a broader introduction to the field, and 
proofs of the claims below, we refer the reader  to the excellent monograph of 
\cite{CoverT91}.

For two distributions $\mu$ and $\nu$ in the same probability space, the 
\it Kullback-Leiber \rm divergence between $\mu$ and $\nu$
is defined as
\begin{align} \label{def_mi_div}
\Div{\mu(a)}{\nu(a)} := \E_{a\sim \mu}\left[\log \frac{\mu(a)}{\nu(a)}\right].
\end{align}


The following well known inequality upper bounds the statistical distance between two distributions in terms of their KL Divergence:
\begin{lemma}[Pinsker's inequality] \label{lemma:pinsker} 
For any two distributions $\mu$ and $\nu$, \[ |\mu(a)- \nu(a)|^2 \leq \dfrac{1}{2} \cdot \Div{\mu(a)}{\nu(a)}. \]
\end{lemma}

A related measure which is central to this paper is that of \emph{mutual information}, 
which captures correlation between  random variables.

\begin{definition}[Conditional Mutual Information]\label{def_MI_div}
Let $A, B, C$ be jointly distributed random variables. The \emph{Mutual Information} between $A$ and $B$
conditioned on $C$ is 
\begin{align*}
& I(A;B|C) :=  
 \Ex{\mu(cb)}{ \Div{\mu(a|bc)}{\mu(a|c)}}  
= \Ex{\mu(ca)}{ \Div{\mu(b|ac)}{\mu(b|c)}} =  \sum_{a,b,c} \mu(abc) \log \frac{\mu(a|bc)}{\mu(a|c)}.
\end{align*} 
\end{definition}

The above definition can be interpreted as follows: $I(A;B|C)$ is large if the distribution
$(A|B=b, C=c)$ is ``far" from $(A|C=c)$  for typical values of $b,c$, which means that $B$ provides a lot of information about $A$ 
conditioned on $C$.
We note that an equivalent, more intuitive definition of (conditional) mutual information is $ I(A;B|C) = H(A|C)-H(A|BC)$,
where $H(A|C)$ is the (expected) \emph{Shannon Entropy} of the random variable $A$ conditioned on $C$. 
Thus $A$ and $B$ have large mutual information conditioned on $C$, if further conditioning on $B$ significantly 
reduces the entropy of $A$. We prefer Definition \ref{def_MI_div} as it is more appropriate for our proof,  but 
we note that the latter one immediately implies 
\begin{fact} \label{fact_mi_bounded_by_entropy}
$I(A;C|D) \leq H(A|D) \leq H(A) \leq |A|$,
\end{fact}
where the last term denotes the cardinality of $\log |\supp(A)|$ of the random variable $A$, and 
the second transition follows since conditioning never increases entropy.

The most important property of mutual information is that it satisfies the following chain rule:
\begin{fact}[Chain rule for mutual information] \label{fact_chain_rule}
Let $A, B, C, D$ be jointly distributed random variables. Then
$I(AB; C|D) = I(A; C|D) + I(B; C|AD)$.
\end{fact}

\begin{lemma}[Conditioning on independent variables increases information] \label{lem_cond_indep}
Let $A,B,C,D$ be jointly distributed random variables. If $I(A;D|C)=0$, 
then it holds that 
$I(A;B|C) \leq I(A;B|CD)$.
\end{lemma}

\begin{proof}
We apply the chain rule twice. On one hand, we have
$ I(A;BD|C) = I(A;B|C) + I(A;D|CB) \geq I(A;B|C)$, 
since mutual information is nonnegative. On the other hand,
$ I(A;BD|C) = I(A;D|C) + I(A;B|CD) = I(A;B|CD)$, 
since $I(A;D|C)=0$ by assumption. Combining both equations completes the proof.
\end{proof}

\vspace{0.1in}

On the other hand, the following lemma asserts a condition under which conditioning \emph{decreases} information:
\begin{lemma} \label{lem_cond_decreases_info}
Let $A,B,C,D$ be jointly distributed random variables such that $I(B;D|AC)=0$. Then 
it holds that 
$I(A;B|C) \geq I(A;B|CD)$.
\end{lemma}

\begin{proof}
Once again, we apply the chain rule twice. We have
$I(A; B|CD) = I(AD; B|C) - I(D; B|C) = I(A; B|C) +  I(D; B|AC) - I(D; B|C)
= I(A; B|C) -  I(D; B|C) \leq I(A; B|C).$
\end{proof}



\begin{fact}[Data processing inequality, general case]\label{fact_dp_ineq_general}
Let $X \rightarrow Y \rightarrow Z$ be a Markov chain ($I(X;Z|Y)=0$). Then $I(X;Z) \leq I(X;Y)$.
\end{fact}

\begin{fact}[Data processing inequality, special case]\label{fact_dp_ineq}
Let $A,B,C$ be three jointly distributed random variables, where the domain of $B$ is $\Omega$, and let $f:\Omega\longrightarrow \cU$
be any deterministic function. Then  $I(A;B|C) \geq I(A;f(B)|C)$.
\end{fact}

\begin{fact}\label{fact_expectation_l_1}
Let $\mu$ and $\nu$ be two probability distributions over a non-negative random variable $X$, whose value is bounded by $X_{max}$. 
Then $\Ex{\nu}[X] \leq  \Ex{\mu}[X] + |\mu-\nu|\cdot X_{max}.$
\end{fact}


\section{A hard distribution for $\large\lowercase{r}$-round protocols}
We begin by defining a family of hard distributions for protocols with $r$ rounds.
Recall that $\cG(n,m)$ is the family of bipartite graphs on $(n,m)$ vertex-sets.
For any given number of rounds $r$, we define a hard distribution $\mu_r$ on bipartite graphs in $\cG(n_r,m_r)$.
$\mu_r$ is recursively defined in Figure \ref{figure:mu_r}.\\ 

\vspace{.2in}

\begin{figure}[H]
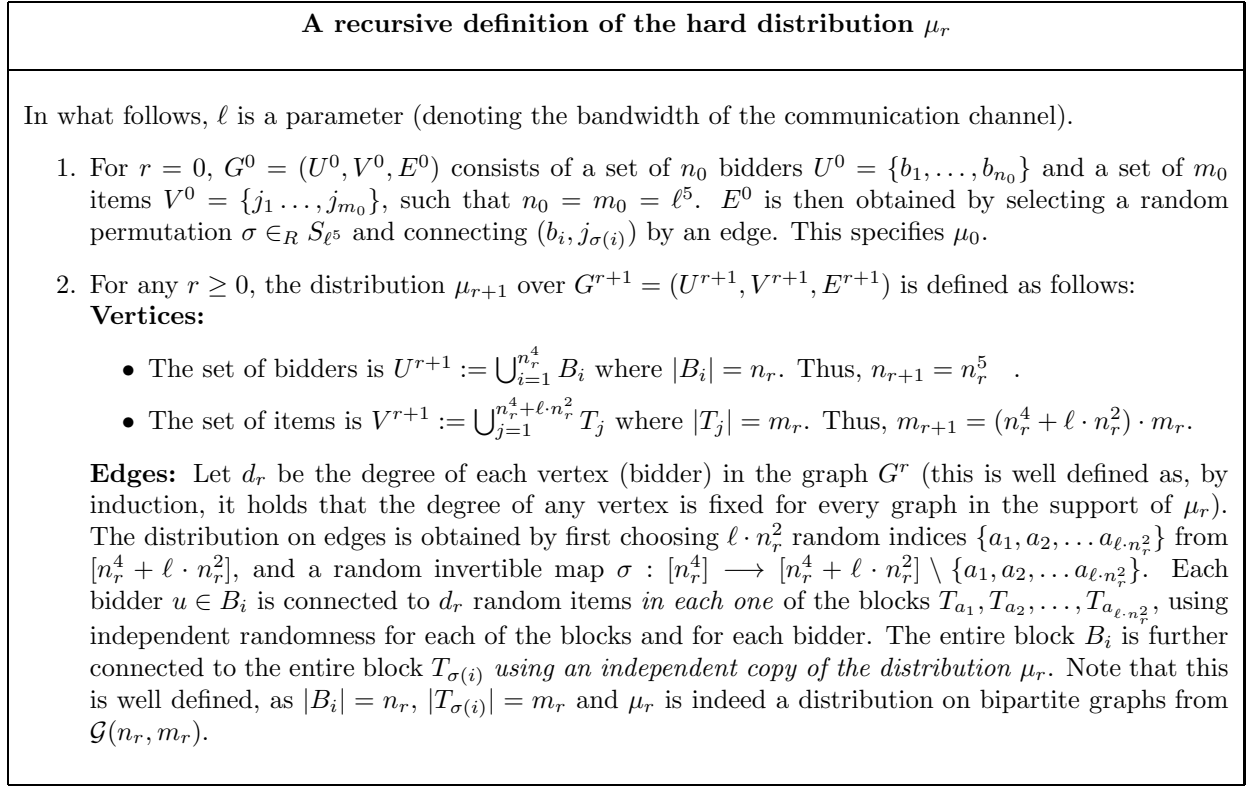

\begin{tabular}{|l|}
\hline
\begin{minipage}{\algwidth}
\vspace{1ex}
\begin{center}
\textbf{A recursive definition of the hard distribution $\mu_r$}
\end{center}
\vspace{0.5ex}
\end{minipage}\\
\hline
\begin{minipage}[t]{\algwidth}
\vspace{1ex}

In what follows, $\ell$ is a parameter (denoting the bandwidth of the communication channel). 
\begin{enumerate}
\item For $r=0$, $G^0=(U^0,V^0,E^0)$ consists of a set of $n_0$ bidders $U^0 = \{b_1,\ldots, b_{n_0}\}$ and a set of $m_0$ items 
$V^0 = \{j_1\ldots , j_{m_0}\}$, such that
$n_0 = m_0 = \ell^{5}$. $E^0$ is then obtained by selecting a random permutation $\sig \in_R S_{\ell^{5}}$ and connecting 
$(b_i,j_{\sig(i)})$ by an edge. This specifies $\mu_0$.
\item For any $r\geq 0$, the distribution $\mu_{r+1}$ over $G^{r+1}=(U^{r+1},V^{r+1},E^{r+1})$ is defined as follows: \\
\bf Vertices: \rm  
\begin{itemize}
\item  The set of bidders is $U^{r+1} := \bigcup_{i=1}^{n_r^4} B_i$ where $|B_i|=n_{r}$. Thus, $n_{r+1} = n_r^{5}$ \; .
\item The set of items is $V^{r+1} := \bigcup_{j=1}^{n_r^4 + \fs} T_j$ 
where $|T_j|=m_{r}$. Thus, $m_{r+1} =  (n_r^4 + \fs)\cdot m_r$.
\end{itemize}
\bf Edges: \rm  
Let $d_r$ be the degree of each vertex (bidder) in the graph $G^r$ (this is well defined as, by induction, it holds that the degree of any vertex 
is fixed for every graph in the support of $\mu_r$).
The distribution on edges is obtained by first choosing $\fs$ random indices $\{a_1,a_2,\ldots a_{\fs} \}$ from 
$[n_r^4 + \fs]$, and a random invertible map $\sig : [n_r^4] \longrightarrow [n_r^4 + \fs]\setminus \{a_1,a_2,\ldots a_{\fs} \}$.
Each bidder $u \in B_i$ 
is connected to $d_r$ random items \emph{in each one} of the blocks $T_{a_1}, T_{a_2},\ldots , T_{a_{\fs}}$, using independent
randomness for each  of the blocks and for each bidder.
The entire block $B_i$ is further connected to the entire block $T_{\sig(i)}$  \emph{using an
independent copy of the distribution $\mu_r$}.
Note that this is well defined, as $|B_i|=n_{r}$, $|T_{\sig(i)}|=m_{r}$ and $\mu_{r}$ is indeed a distribution 
on bipartite graphs from $\cG(n_{r},m_{r})$. 

\end{enumerate}

\vspace{1.5ex}
\end{minipage}\\
\hline
\end{tabular}
\caption{A hard distribution for $r$-round protocols.}
\label{figure:mu_r}
\end{figure}

\begin{remark} \label{rem_properties}
A few remarks are in order:
\begin{enumerate}[(i)]
\item  As standard, the input of each bidder $u \in U^{r+1}$ is the set of incident edges on the vertex $u$ (defined by $\mu_{r+1}$). 
Note that every graph in the support of $\mu_{r+1}$ has a perfect matching ($|\cM(G^{r+1})| = n_{r+1}$). 
\item It is easy to see by induction that : (a) $n_r=\ell^{5^{r+1}}$; and (b) $m_r \leq n^2_r$. \\
(Proof of (b):  By induction on $r$, $m_{r+1} := (n_r^4 + \fs)\cdot m_r \leq (n_r^4 + \fs)\cdot n^2_r \leq 2n^6_r < n^2_{r+1}$).
\item  Note that in $\mu_{r+1}$, each block of bidders $B_i$ is connected to its ``hidden item block" $T_{\sig(i)}$ using a copy of the joint 
distribution $\mu_r$, and to each of the ``fooling item blocks" $T_{a_j}$, using the \emph{product of the marginals} of $\mu_r$, i.e., according to 
$\bigtimes_{u\in n_r} (\mu_r|u)$. This property will be crucial. 
\item Throughout the paper, we assume the bandwidth parameter $\ell$ is larger than some large enough absolute constant (note that by (ii) above, in fact
$\ell =\omega_r(1)$).
\end{enumerate}
\end{remark}

\paragraph{\bf Notation. \rm}

To facilitate our analysis, the following notation will be useful.
Notice that each block $B_i$ of players is connected to exactly $\fs+1$ blocks of items  whose indices we denote by
\[ \cI_i := \{\sig(i),a_1,a_2,\ldots a_{\fs} \}. \]
For each $B_i$, let $\tau_i : \cI_i \longrightarrow [\fs+1]$ be the bijection that maps any index in $\cI_i$
to its location in the sorted list of $\cI_i$ (i.e., $\tau_i^{-1}(1)$ is the smallest index in $\cI_i$, $\tau_i^{-1}(2)$ is the 
second smallest index in $\cI_i$ and so forth).  
We henceforth denote by $G^i_j$ the (induced) subgraph of $G=G^{r+1}$ on the sets $(B_i,T_{\tau_i^{-1}(j)})$,
for each $j\in [\fs+1]$. By a slight abuse of notation, we will sometimes write $G^i_j=(B_i,T_{\tau_i^{-1}(j)})$ to denote the specific set of edges 
of $G^i_j$.
Similarly, for a bidder $u\in B_i$, let $G^u_j=(u,T_{\tau_i^{-1}(j)})$ denote the (induced) subgraph of $G$ on the sets $(u,T_{\tau_i^{-1}(j)})$.
In this notation, the entire input of a player $u\in B_i$ is $\Gamma_u := \{G^u_1,G^u_2,\ldots , G^u_{\fs+1}\}.$ 
Let $$J_i := \tau_i(\sig(i))$$ denote the index of the ``hidden graph" 
$G^i_{J_i} = (B_i,T_{\sig(i)})$. To avoid confusion (with the other indices $j$), we henceforth write $$G(J_i):=G^i_{J_i}.$$ 
Note that by symmetry of our construction, the index $J_i$ is uniformly distributed in $[\fs+1]$. 
The following fact will be crucial to our analysis:

\begin{fact}[Marginal Indistinguishability]\label{fact_marginal_indistinguishability}
For any bidder $u\in B_i$, it holds that $I(\Gamma_u;J_i \; | \; \cI_i)=0.$
\end{fact}

\begin{proof}
Recall that $\Gamma_u = \{G^u_1,G^u_2,\ldots , G^u_{\fs+1}\}$ is the input of bidder $u$. The claim follows directly from property $(iii)$
in Remark \ref{rem_properties}, since 
by definition of our construction, the distribution of edges of $G^u_j = (u,T_{\tau_i^{-1}(j)})$ is $(\mu_{r}|u)$ for all $j\in  [\fs+1]$.
We remark that the above fact implies that, up to a permutation on the names of the items in $V^{r+1}$, $G^u_j \sim G^u_k$ 
for any bidder $u \in B_i$ and any $j\neq k \in [\fs+1]$. 
\end{proof}

Finally, Let  $\cB$ denote the partition of bidders in $U:=U^{r+1}$ into the blocks $B_i$,
and $\cT$ denote the partition of items in $V:=V^{r+1}$ into the blocks $T_j$. 
Throughout the proof, we think of $\cT$ and $\cB$ as fixed, while we think of the names of the bidders in each block of $\cB$ 
and items in each block of $\cT$ as random.
Since $\cT$ and $\cB$ are fixed (publicly known) 
in the distribution $\mu_{r+1}$,
our entire analysis is performed  under the implicit conditioning on $\cT,\cB$.
Note that $\cT$ does not reveal the identity of the ``fooling blocks" $T_{a_j}$, but only the items belonging to each block.


\vspace{0.1in}\section{The lower bound} \label{sec_main_lb}

In this section we prove our main result. Recall that the expected matching size of $\pi$ (with respect to $\mu$) 
is $\E_\mu[|\hat{\cM}(\Pi) \cap E(G)|]$. We shall prove the following theorem.

\begin{theorem}[Main Result] \label{thm_main_lb}
The expected matching size of any $r$-round protocol under $\mu_{r}$ is at most $5n_r^{1-1/5^{r+1}}$. 
This holds as long as the number of bits sent by each player at any round is at most  $\ell = n_r^{1/5^{r+1}}$. 
In particular, since $\mu_r$ has a perfect matching, the approximation ratio of any $r$-round protocol is no better than $\Omega\left( n^{1/5^{r+1}} \right)$. 
\end{theorem}

The intuition behind the proof is as follows.
Consider some $(r+1)$-round protocol $\pi$ (with bandwidth $\ell$), and let $\Mbi = \MbiAll$ denote the (concatenated) messages sent 
by all of the bidders in a block $B_i$ in the first round of $\pi$. From this point on, we will assume that $\pi$ is a deterministic protocol (since by the 
averaging principle we may fix its randomness without harming the performance). 
Informally speaking, the distribution $\mu_{r+1}$ is designed so that messages of bidders in $B_i$ ($\Mbiu$) convey 
little information about the ``hidden" graph $\GiJ$. Intuitively, this will be true since the \emph{marginal} distribution of the hidden graph $G^u_{J_i}$
for any bidder $u\in B_i$
is indistinguishable from the rest of the ``fooling graphs"  
(Fact \ref{fact_marginal_indistinguishability}) and therefore a bidder in $B_i$ will not be able to distinguish between vertices (items)
in $\bigcup_{j=1}^{\fs} T_{a_j}$ and in $\Tsi$. Using the conditional independence properties of the distribution $\mu_{r+1}$ and the simultaneity of the protocol, 
we will show that the latter condition also implies that the \emph{total} information 
conveyed by $\Mbi$ on $G(J_i)$ is small.  In order to make this  information $\ll 1$ bit, the parameters are chosen so that $n_r$ grows 
doubly-exponentially in $r$ ($n_r = \ell^{5^{r+1}}$), and this choice is the cause for the approximation ratio we eventually obtain.
Intuitively, the fact that little information is conveyed by each block on the ``hidden graph" implies that the distribution of edges in the graph 
$\GiJ$  is still close to $\mu_{r}$ even \emph{conditioned} on the first message of the $i$'th block $\Mbi$. 
Now suppose an $(r+1)$-round protocol finds a large matching with respect to the original distribution $\mu_{r+1}$ (in expectation). 
Then the expected induced matching size on $\GiJ$ must be large on average as well. Hence, ``ignoring" the first round of the protocol, 
we would like to argue that the original protocol 
essentially induces an $r$-round protocol for  finding a large matching with respect to the distribution $\mu_r$, up to some error term (indeed, \emph{some} 
information about $\GiJ$ may have already been discovered in the first round of the protocol, but the argument above ensures that this information is small). 
Doing so essentially reduced the problem to finding a large matching under $\mu_r$ using only $r$ rounds, so we may use an inductive approach to 
upper bound the latter expected matching size.  

Making the latter intuition precise is complicated by the fact that, unlike standard ``round-elimination" arguments in the two-party
setting, in our setup one cannot simply ``project" an $r$-round $n_{r+1}$-party protocol (with inputs $\sim \mu_{r+1}$) directly to the 
distribution $\mu_r$, since a protocol for the latter distribution has only $n_r$ players (inputs). To remedy this, we crucially rely on the conditional 
independence properties of our construction (Lemma \ref{prop_m_1_restricted_decomposition} below) together with an embedding argument to 
obtain the desired lower bound.

The embedding part of the proof (Claim \ref{claim_round_elim}) is subtle, since in general, conditioning on the first message $M_1$ correlates the 
(private) inputs of the players with the ``missing" inputs to the ``higher-dimensional" protocol (the ``fooling item blocks" of $\mu_{r+1}$), so it is not 
clear how the players can sample these ``missing" inputs without communicating.
Luckily and crucially, the edges to the ``fooling blocks" $T_{a_j}$ in $\mu_{r+1}$ were chosen independently for each bidder $u\in U$ (unlike the hidden 
graphs $G(J_i)$ in which players have correlated edges). This independence is what allows to embed a lower-dimensional graph $H\sim \mu_r$ 
and ``complete" the rest of the graph (using a combination of public and private randomness) according to the conditional distribution $(G| M_1, H)$ without 
any communication, thus ``saving" one round of communication.

We now turn to formalize the  above intuition. From this point on, let us use the shorthands 
$$ \J := \Js \;\; ,\;\; \cI := \cI_1,\cI_2,\ldots , \cI_{n^4_r}. $$
Also, for the remainder of the proof, let us define for simplicity
\[ \dr := \frac{1}{n_r} \;. \]
Let $\pi$ be an $(r+1)$-round deterministic protocol. For a given message $m_{B_i}:=\mbiAll$ sent in $\pi$ by the bidders in block $B_i$ in the first round of $\pi$, a 
fixing of the index $J_i=j_i$ of the ``hidden" block of items, and of the partition $\cI_i$,  let  
\[\psir := (\GiJ \mid M_{B_i}= m_{B_i} , J_i =  j_i , \cI_i) \] 
denote the distribution of the ``hidden graph" $G(J_i)$ conditioned on $M_{B_i}, \cI_i$ and $J_i$. The following lemma asserts that, in expectation over the first communication round of $\pi$, the marginal distribution of $\GiJ$ is very close to its original distribution $\mu_r$.

\begin{lemma} \label{cl_psi_delta_close}
For every $i\in[n_r^4]$, 
$$\Ex{\mbi,\cI_i,j_i}\left[ | \psir - \mu_r| \right] \leq \dr^{1/2}. $$
\end{lemma}

\begin{proof}
We begin by showing that the local message $\Mbiu$ of any bidder $u\in B_i$ conveys little information on $G(J_i)$. 
Note that Fact \ref{fact_marginal_indistinguishability} (and the Data Processing inequality (Fact \ref{fact_dp_ineq})) together imply that, 
for any block $B_i$ of bidders 
and any $u\in B_i$,
\begin{align}\label{eq_m_indep_sig_i}
I(\Mbiu;J_i \; | \; \cI_i)=0.
\end{align} 
(Note that in contrast, $I(M_{B_i};J_i \; | \;  \cI_i) \neq 0$. In fact, $J_i$ may be almost determined by the 
entire message of the $i$'th block (let alone by the entire message $M_1$ of $\pi$), as the induced distribution of $G(J_i)$ is different
than that of $G^i_{j}$, $j\neq J_i$.
This is where we crucially use the \emph{simultaneaty} of bidder's messages). 
We will also need the following proposition:

\begin{proposition}\label{prop_remove_j}
For any bidder $u\in B_i$ and any $j\in [\fs+1]$, it holds that $$I(\Mbiu ; G^i_j \; |  \cI_i, J_i = j) \leq I(\Mbiu ; G^u_j \; |  \cI_i).$$
\end{proposition}
\begin{proof}
Recall that $\Gamma_u = \{G^u_1,G^u_2,\ldots , G^u_{\fs+1}\}$ is the input of bidder $u$, and notice that for any $j\in [\fs+1]$,
$$(\Mbiu |\cI_i,J_i=j) \rightarrow (\Gamma_u|\cI_i,J_i=j) \rightarrow (G^u_j|\cI_i,J_i=j) \rightarrow (G^i_j|\cI_i,J_i=j)$$ 
is a Markov chain where the left chain holds since, conditioned on $(\Gamma_u,\cI_i,J_i=j)$, $\Mbiu$ is completely determined and therefore 
independent of $G^i_j$ and $G^u_j$, and 
the right chain holds since conditioned on $\cI_i,J_i=j$ and 
the graph $G^u_j$, the rest of the edges of the graph $G^i_j$ are independent of $\Gamma_u$ by construction.
Therefore, by the (general)
Data Processing inequality (Fact \ref{fact_dp_ineq_general}), we have 
\begin{align}\label{eq_remove_gij}
I(\Mbiu ; G^i_j \; |  \cI_i, J_i = j) \leq I(\Mbiu ; G^u_j \; |  \cI_i, J_i = j).
\end{align} 
Now, by Fact \ref{fact_marginal_indistinguishability}, we know that the distribution of $(\Gamma_u|\cI_i)$ is independent of the event $``J_i=j"$.
Since $\Mbiu$ and $G^u_j$ are deterministic functions of $\Gamma_u$ (conditioned on $\cI_i$), this also implies that the joint distribution 
of $(\Mbiu,G^u_j | \cI_i)$ is independent of the event $``J_i=j"$. Therefore, we conclude by \eqref{eq_remove_gij} that 
\begin{align*}
I(\Mbiu ; G^i_j \; |  \cI_i, J_i = j) \leq I(\Mbiu ; G^u_j \; |  \cI_i, J_i = j) = I(\Mbiu ; G^u_j \; |  \cI_i).
\end{align*} 
\end{proof}
We proceed to prove the Lemma. We may now write for any $u\in B_i$
\begin{align}
&I(\Mbiu ; \GiJ \; | J_i, \cI_i) 
 = \frac{1}{\fs+1}\cdot \sum_{j=1}^{\fs+1} I(\Mbiu ; G^i_j \; |  \cI_i, J_i = j)  \nonumber \\
& \text{(By definition of conditional mutual information and since $J_i \in_R [\fs+1]$ and by \eqref{eq_m_indep_sig_i})} \nonumber \\
& \leq \frac{1}{\fs+1}\cdot \sum_{j=1}^{\fs+1} I(\Mbiu ; G^u_j \; |  \cI_i)  \;\;\;\;\;\;\;\; \text{(By Proposition \ref{prop_remove_j})} \nonumber \\
& \leq \frac{1}{\fs+1}\cdot \sum_{j=1}^{\fs+1} I(\Mbiu ; G^u_j \; | G^u_1,G^u_2,\ldots, G^u_{j-1},   \cI_i) \label{eq_more_conditioning}\\
& = \frac{1}{\fs+1}\cdot I(\Mbiu ; G^u_1,G^u_2,\ldots, G^u_{\fs+1} | \cI_i)  \;\;\;\;\; \text{(by the chain rule)} \nonumber \\
&\leq \frac{H(\Mbiu)}{\fs+1} \;\;\;\;\;\;\;\;\;\;\;\;\; \text{(by Fact \ref{fact_mi_bounded_by_entropy})} \nonumber \\
& \leq \frac{|\Mbiu|}{\fs+1} \leq \frac{\ell}{ \fs+1} < \frac{1}{n^2_r} = \dr^2  \label{eq_low_info}
\end{align}
where the inequality in \eqref{eq_more_conditioning} follows from Lemma \ref{lem_cond_indep} taken with 
$A=G^u_j, B=\Mbiu,C=\cI_i, D=G^u_{<j}$, since $G^u_j$ is independent of $G^u_{<j}$ for all $j$, conditioned on $\mathcal{I}_i$.

Now, we claim that, for each bidder $u\in B_i$, conditioning on the previous  messages of the bidders 
($M^{<u}_{B_i} := M^{1}_{B_i} M^{2}_{B_i}  \ldots M^{u-1}_{B_i} $) can only \emph{decrease} the information $\Mbiu$ reveals on 
the hidden graph $G(J_i)$:
\begin{claim} \label{cl_cond_prev_messages_decrease_info}
$I(\Mbiu ; \GiJ \; | M^{< u}_{B_i} , J_i, \cI_i) \leq I(\Mbiu ; \GiJ \; | J_i, \cI_i).$
\end{claim}

\begin{proof}
By construction of $\mu_{r+1}$, conditioned on $G(J_i),\cI_i$ and $J_i$, the inputs of bidders $u$ and $B_i\setminus \{u\}$ are \emph{independent}.
In particular, this fact and the data processing  inequality (Fact \ref{fact_dp_ineq}) together imply that
\[ I(\Mbiu ; M^{< u}_{B_i} \; | \GiJ , J_i, \cI_i)  = 0,\]
since $\pi$ was assumed to be a deterministic protocol.
By non-negativity of information and the chain rule, 
\begin{align*}
&I(\Mbiu ; \GiJ \; | M^{< u}_{B_i} , J_i, \cI_i) \leq I(\Mbiu ; \GiJ , M^{< u}_{B_i} \; | J_i, \cI_i) \\
& = I(\Mbiu ; \GiJ \; | J_i, \cI_i) + I(\Mbiu ; M^{< u}_{B_i} \; | \GiJ , J_i, \cI_i) \\
& = I(\Mbiu ; \GiJ \; | J_i, \cI_i).
\end{align*}
\end{proof}

We conclude that
\begin{align} \label{eq_total_info_bound}
& \Ex{\mbi, j_i,\cI_i}\left[ \Div{\psir}{\mu_r} \right]  = I(\Mbi ; \GiJ \; | J_i, \cI_i)  \;\;\;\; (\text{by Definition \ref{def_MI_div} of conditional mutual information}) \nonumber \\
& = \sum_{u\in B_i} I(\Mbiu ; \GiJ \; | M^{< u}_{B_i} , J_i, \cI_i)  \;\;\;\;\;\; (\text{by the chain rule}) \nonumber \\
&\leq \sum_{u\in B_i} I(\Mbiu ; \GiJ \; | J_i, \cI_i)  \;\;\;\;\; \text{(by Claim \ref{cl_cond_prev_messages_decrease_info})} \nonumber \\
& \leq |B_i|\cdot \dr^2  \;\;\;\; \text{(by \eqref{eq_low_info})} \nonumber \\
& = n_r \cdot \dr^2 = \dr.
\end{align}

Combining \eqref{eq_total_info_bound}, Pinsker's inequality (Lemma \ref{lemma:pinsker}) and convexity of $\sqrt{\cdot}$ completes the entire proof of the lemma.

\end{proof}

\vspace{.1in}
We are now ready to prove Theorem \ref{thm_main_lb}. To this end, for any $r$-round protocol $\pi$, input graph 
$G$, and induced subgraph $H\subseteq G$,  
let $$N_{\pi}(G,H) := |\hat{\cM}(\Pi(G)) \cap E(H)|$$ 
denote the size of the matching computed from $\pi$'s transcript with 
respect to the 
subgraph $H$ (note that $N_{\pi}(G,H)$ is a random variable depending on $G$). 
For notational convenience, we use the shorthand $N_{\pi}(G) := N_{\pi}(G,G)$. 
Theorem \ref{thm_main_lb} will follow directly from the following theorem:

\vspace{.1in}

\begin{theorem}\label{thm_main_lower_bound}
Let $\pi$ be an $r$-round (deterministic) communication protocol with bandwidth $\ell$. 
Then
\[ \Ex{G\sim \mu_{r}} [N_{\pi}(G)] \leq \tr .\]
\end{theorem}

\vspace{.1in} 

\begin{proof}
We prove the theorem by induction on $r$. 
Let us denote $$t(r):=\tr. $$ 
For $r=0$ (namely, with no communication at all), 
the expected number of edges the referee guesses correctly under $\mu_0$ is at most 
$ n_0\cdot\frac1 {n_0} = 1 =  t(0)$ (as $G^0\sim \mu_0$ is a random permutation on $[n_0]$).

Suppose the theorem statement holds for all integers up to $r$. Thus, the expected matching produced by 
any $r$-round protocol $\theta$ (with bandwidth $\leq \ell$) under $\mu_r$ satisfies
\begin{align}\label{eq_induction_hyp}
\Ex{G\sim \mu_{r}} [N_{\theta}(G)] \leq t(r) . 
\end{align}
We need to show that the expected matching produced by any $(r+1)$-round protocol $\pi$ (with bandwidth $\leq \ell$) under $\mu_{r+1}$ satisfies
\begin{align}\label{eq_induction_hyp}
\Ex{G\sim \mu_{r+1}} [N_{\pi}(G)] \leq t(r+1) . 
\end{align}

Let $\pi$ be an $(r+1)$-round protocol.
Recall that $G\sim \mu_{r+1}$ consists of $n_r^4$ ``blocks" $B_i$ of bidders, each of which is connected to exactly $|\cI_i|=\fs+1$ item blocks. 
Let $M_1 := M_{B_1} M_{B_2}\ldots M_{B_{n_r^4}}$ denote the 
messages sent by each \emph{block} of bidders in the first round of $\pi$ (where $\Mbi = M^1_{B_i}, M^2_{B_i},\ldots , M^{n_r}_{B_i}$ is the concatenated 
message of all bidders $u\in B_i$). 
Recall that for every $i\in [n^4_{r}]$, 
$G(J_i)$ denotes the induced subgraph of $G$ on $(B_i,T_{\tau_i^{-1}(J_i)})$, and that for every bidder $u\in B_i$,
$G^u_{J_i} = (u,T_{\tau_i^{-1}(J_i)})$ denotes the induced subgraph between bidder $u$ and the ``hidden graph" of the $i$'th block to which $u$ belongs. 
In the same spirit, for every block $B_i$ and every bidder $u\in B_i$, let 
$$G(T_i) :=\left(B_i, \bigcup_{j=1}^{\fs} T_{a_j}\right) \;\;\;\;\; , \;\;\;\;\; \Gtu :=\left(u, \bigcup_{j=1}^{\fs} T_{a_j}\right)$$
denote the induced subgraph on the block $B_i$ (on the bidder $u\in B_i$) and all ``fooling blocks" respectively.
As usual, for any subset $S\subseteq [n^4_r]$,  
we write $G(T_S) :=\left(\bigcup_{i\in S} B_i, \bigcup_{j=1}^{\fs} T_{a_j}\right)$ and use the convention $\T := T_{[n^4_r]}$. 
In what follows, $G(\J) := G(J_1)G(J_2)\ldots G(J_{n^4_r})$ denotes the (concatenation of the) ``hidden" graphs.  
The following proposition will be essential for the rest of our argument:

\begin{lemma}[Conditional Subgraph Decomposition] \label{prop_m_1_restricted_decomposition}
The following conditions hold: 
\begin{enumerate} 
\item  $ ((G^1_T, G^2_T, \ldots , G^{n_r}_T) \; | \; M_1, G(J_1), \J, \cI) \; \sim \;  \bigtimes_{u\in B_1} (\Gtu | M_1, \Gju, \J, \cI)$, \\
where $\{1,2,\ldots , n_r\}$ are the bidders of the first block $B_1$.
\item 
$\left(G(\J),G(\T) \; | \;  M_1,\J,\cI \right) \sim \bigtimes_{i\in [n_r^4]} (G(J_i)G(T_i) \; |\; \Mbi,\J, \cI).$
\end{enumerate}
That is, the joint distribution of the ``fooling subgraphs" $G^u_T$ of each bidder $u\in B_1$  conditioned on the entire message $M_1$
and the ``hidden graph" $G(J_1)$ of the first block,  is a product of the marginal distributions $G^u_T$ conditioned only on the 
``local hidden part" $\Gju$ of each bidder and 
$M_1$.

Furthermore, the joint distribution of the subgraphs induced on each block ($G(J_i)G(T_i))$  conditioned on the entire message $M_1$ 
is a product distribution of the marginal distributions of the $i$'th block, conditioned only on the ``local" message of the respective block 
$M_{B_i}$ (In particular, these graphs remain independent even conditioned on $M_1$).
\end{lemma}

The intuition behind the second proposition is clear: Since in the original distribution $\mu_{r+1}$, the graphs of each block are independent 
by construction, this remains true even when conditioned on the first (deterministic) message of each block. The first proposition is more
subtle, since within the same block (say $B_1$), the inputs of the bidders $u\in B_1$ are correlated (via the hidden graph $G(J_1)$).
However, conditioned on knowing the hidden block ($\J,\cI$), the marginal distribution of  the
fooling graph $\Gtu$ is \emph{independent} for each $u$ by construction, and therefore the only correlation between $G(J_1)$ and
 $\Gtu$ created by conditioning on the message $M_1$, is correlation between the ``local hidden graph" of bidder u ($\Gju$)
and $\Gtu$. We remark that this fact will be used crucially in the embedding argument below (Claim \ref{claim_round_elim}).
We proceed to the formal proof.

\begin{proof}[Proof of Lemma \ref{prop_m_1_restricted_decomposition}]
We repeatedly use Lemma \ref{lem_cond_decreases_info}.

\paragraph{\bf Proof of (1) \rm}
It suffices to show that for every $u\in B_1$, $I(\Gtu ; \Gtmu \Gjomu | M_1, \J, \cI, \Gju) = 0$.
To this end, observe that 
\begin{align} \label{eq_decomp_1}
I(\Gtu ; M^{-u}_1 | M^u_1, G(J_1), \Gtmu , \J, \cI) \leq H(M^{-u}_1 | G(J_1), \Gtmu , \J, \cI)  = 0, 
\end{align}
since the message $M^{-u}_1$ of all bidders in $B_1$ except bidder $u$ is fully determined by the inputs $(G(J_1), \Gtmu)$. 
For the same reason,
\begin{align} \label{eq_decomp_2}
I(\Gtmu \Gjomu ; M^u_1| \Gju, \Gtu, \J, \cI) \leq  H(M^u_1| \Gju, \Gtu, \J, \cI) = 0. 
\end{align}
Therefore,
\begin{align*}
& I(\Gtu;\Gtmu \Gjomu | M_1, \J, \cI, \Gju) = I(\Gtu;\Gtmu \Gjomu | M^u_1, M^{-u}_1, \J, \cI, \Gju)  \\
 \leq \; &  I(\Gtu;\Gtmu \Gjomu | M^u_1, \J, \cI, \Gju) \;\;\;\; \text{(By Lemma \ref{lem_cond_decreases_info} with $D=M^{-u}_1$, and \eqref{eq_decomp_1})} \\ 
 \leq \; & I(\Gtu;\Gtmu \Gjomu | \J, \cI, \Gju) \;\;\;\;\;\;\; \text{(By Lemma \ref{lem_cond_decreases_info} with $D=M^{u}_1$, and \eqref{eq_decomp_2})} \\
 = \; & 0, \;\;\;\; \text{as desired.}
\end{align*}

\paragraph{\bf Proof of (2) \rm}
It suffices to show $I(G(J_i)G(T_i) ; G(J_{-i})G(T_{-i}) M_{B_{-i}} | M_{B_i}, \J, \cI) = 0$.
Once again, applying Lemma \ref{lem_cond_decreases_info} with $D=M_{B_i}$, we have
\begin{align} \label{eq_decomp_3}
I(G(J_i)G(T_i) ; G(J_{-i})G(T_{-i})  | M_{B_i}, \J, \cI) \leq I(G(J_i)G(T_i) ; G(J_{-i})G(T_{-i})  | \J, \cI)
\end{align}
since $I( M_{B_i} ; G(J_{-i})G(T_{-i}) | G(J_i), G(T_i), \J, \cI) \leq H( M_{B_i} | G(J_i),G(T_i), \J, \cI ) = 0$ where
the last equality is because $\Mbi$ is determined by the input of block $B_i$.
The same argument implies 
\begin{align} \label{eq_decomp_4}
I(G(J_i)G(T_i) ; M_{B_{-i}} | M_{B_i}, \J, \cI, G(J_{-i})G(T_{-i})) \leq I(G(J_i)G(T_i) ; M_{B_{-i}} | \J, \cI, G(J_{-i})G(T_{-i}))
\end{align}
since once again, $I(M_{B_i} ; M_{B_{-i}} | G(\J),G(\T), \J, \cI) \leq H(M_{B_i} | G(\J),G(\T), \J, \cI) = 0$.
Combining equations \eqref{eq_decomp_3} and \eqref{eq_decomp_4}, we conclude by the chain rule that 
\begin{align} \label{eq_decomp_5}
& I(G(J_i)G(T_i) ; G(J_{-i})G(T_{-i}) M_{B_{-i}} | M_{B_i}, \J, \cI) \nonumber \\
= \; & I(G(J_i)G(T_i) ; G(J_{-i})G(T_{-i})  | M_{B_i}, \J, \cI) + I(G(J_i)G(T_i) ; M_{B_{-i}} | M_{B_i}, \J, \cI, G(J_{-i})G(T_{-i})) \nonumber \\
\leq \; & I(G(J_i)G(T_i) ; G(J_{-i})G(T_{-i})  | \J, \cI) + I(G(J_i)G(T_i) ; M_{B_{-i}} | \J, \cI, G(J_{-i})G(T_{-i}) ) \nonumber \\ 
= \; & 0,
\end{align}
where the last transition follows from the definition of $\mu_{r+1}$, and 
since $M_{B_{-i}}$ is a deterministic function of $G(J_{-i})G(T_{-i})$ conditioned on $\J,\cI$.

\end{proof}

\vspace{.2in}

We now proceed to prove \eqref{eq_induction_hyp}, the inductive step of the proof. 
Recall that $\pi$ is assumed to be deterministic, but $M_1$ is still
a random variable (under the input distribution $\mu_{r+1}$). Hence, we may 
equivalently draw $G\sim \mu_{r+1}$ by first sampling the first message $m_1\sim M_1$, and then sampling $G \sim \mu_{r+1}|m_1$. 
Let us denote by $\pi|m_1$ the protocol which is the subtree of $\pi$ conditioned on the first message being $m_1$. 
Note that $\pi|m_1$ has only $r$ rounds of communication. We therefore have

\begin{align}\label{eq_main_1}
& \Ex{G\sim \mu_{r+1}} [N_{\pi}(G)] \leq \Ex{\substack{m_1 \\ \J, \cI}}\;\Ex{\substack{G|m_1,  \J, \cI}} \left[ \ell\cdot n_r^2 \cdot m_r + 
\sum_{i=1}^{n_r^4} N_{\pi|m_1}(G, G(J_i))  \right]  
\end{align}
since any matching in $G$ can at most match all of the items in $\bigcup_{j=1}^{\fs} T_{a_j}$ and each block $T_{a_j}$ contains $m_r$ edges by definition 
of $\mu_{r+1}$, and the rest of the matched edges are contained in $G(J_1), G(J_2), \ldots , G(J_{n_r^4})$. 
Recall that by definition of $\mu_{r+1}$, $G(J_i)\sim \mu_r$.
Hence by linearity of expectation and  the second proposition of Lemma \ref{prop_m_1_restricted_decomposition},  
we may equivalently write the above as
\begin{align}
& = \ell\cdot n_r^2 \cdot m_r +   \sum_{i=1}^{n_r^4} \Ex{\substack{m_1 \\ \J, \cI}} \; 
\Ex{\substack{G(J_i)|(m_{B_i},\J, \cI) \\ G|\; (G(J_i),m_1,\J, \cI)}} \left[N_{\pi|m_1}(G, G(J_i)) \right] \nonumber \\
& = \ell\cdot n_r^2 \cdot m_r +  \sum_{i=1}^{n_r^4} \Ex{\substack{m_1 \\ \J, \cI}} \Ex{\substack{G(J_i)\sim \psir \\ G|\; (G(J_i),m_1,\J, \cI)}} 
\left[ N_{\pi|m_1}(G, G(J_i))  \right], \label{eq_main_2}
\end{align}
by definition of $\psir$ (actually, $\psir$ is defined conditioned only on $J_i,\cI_i$ but conditioning on all indices $\J,\cI$ clearly doesn't change the distribution).
By symmetry of the distribution $\mu_{r+1}$,
it suffices to upper bound the first term in the above summation
$$ \Ex{\substack{m_1 \\ \J, \cI}} \Ex{\substack{G(J_1)\sim \psi^1_r \\ G|\; (G(J_1),m_1,\J, \cI)}} \left[ N_{\pi|m_1}(G, G(J_1))  \right]. $$
To this end, we have 
\begin{align}  
& \Ex{\substack{m_1 \\ \J, \cI}} \Ex{\substack{G(J_1)\sim \psi^1_r \\ G|\; (G(J_1),m_1,\J, \cI)}} \left[ N_{\pi|m_1}(G, G(J_1))  \right]  \nonumber \\
& \leq \Ex{\substack{m_1 \\ \J, \cI}}\left[|\psi^1_r - \mu_r|\right]\cdot n_r 
+ \Ex{\substack{m_1 \\ \J, \cI}} \;\Ex{\substack{G(J_1)\sim \mu_r \\ G|\; (G(J_1),m_1,\J, \cI)}} \left[ N_{\pi|m_1}(G, G(J_1))\right]   \label{eq_main_3} \\
& \leq \dr^{1/2}\cdot n_r 
+ \Ex{\substack{m_1 \\ \J, \cI}} \Ex{\substack{G(J_1)\sim \mu_r \\ G|\; (G(J_1),m_1,\J , \cI)}} \left[ N_{\pi|m_1}(G, G(J_1))  \right]    \label{eq_main_4}
\end{align}
where \eqref{eq_main_3} follows from fact \ref{fact_expectation_l_1}, since 
trivially $N_{\pi|m_1}(G,G(J_1) \leq |U^r| = n_r$ for any $G(J_1)$, 
and the last transition \eqref{eq_main_4}
follows from Lemma \ref{cl_psi_delta_close}.  
We now wish to use the inductive hypothesis to argue that the rightmost term of \eqref{eq_main_4} cannot exceed the expected matching size of an $r$-round
protocol over $\mu_r$. We do so using an embedding argument, which is the heart of the proof.

\begin{claim}[$r$-round Embedding]\label{claim_round_elim}
$$\Ex{\substack{m_1 \\ \J, \cI}} \Ex{\substack{G(J_1)\sim \mu_r \\ G|\; (G(J_1),m_1,\J, \cI)}} \left[ N_{\pi|m_1}(G, G(J_1))  \right]  \leq t(r).$$
\end{claim}
\begin{proof}
Notice that the protocol $\pi|m_1$ is defined over inputs from $\mu_{r+1}$ and not $\mu_r$, so we cannot apply the inductive hypothesis
directly to obtain our desired upper bound. Instead, we will ``embed" $H\sim \mu_r$ into $\pi|m_1$ by simulating the rest of the players using 
public randomness (and then fix the public coins to obtain a deterministic protocol). To this end, for every bidder $u\in U^r$, denote by $H_u$
the induced subgraph of $H$ on the vertex $u$ (i.e., the input of bidder $u$ in $\mu_r$).

Consider the following $r$-round randomized protocol $\tau$ for $H\sim \mu_r$: 
The $n_r$ players  use the shared random tape to sample $(M_1, \J, \cI)$ (according to the probability space of $\pi$). 
Then, they ``embed" their inputs (the graph $H$) to the first block $B_1$
and each bidder $u\in B_1$ ``completes" his missing edges to the ``fooling graph" $\Gtu$ according to $(\Gtu | M_1, H_u, \J, \cI)$ 
\emph{using private randomness}. Note that this is possible due to the first proposition of Lemma  \ref{prop_m_1_restricted_decomposition},
since it asserts that $(\Gtu | M_1, H, \J, \cI) \sim (\Gtu | M_1, H_u, \J, \cI)$.
The players now use the second proposition of Lemma  \ref{prop_m_1_restricted_decomposition} to sample the graphs of the rest of the blocks 
according to $\bigtimes_{i =2}^{n_r^4} (G(J_i)G(T_i) \; |\; \Mbi,\J, \cI)$, as the proposition asserts that 
these subgraphs remain independent after the conditioning. 

This process specifies a graph $G$ such that $H\sim \mu_r$ and $G\sim \mu_{r+1}$ conditioned on $G(J_1)=H,M_1,\J,\cI$.
Notice that so far the players have not communicated at all.
The players now run the $r$-round protocol $\pi|m_1$ and outputs its induced matching on $H$. 
Notice that this protocol is well defined, as the messages of bidders outside block $B_1$ 
in $\pi|m_1$ in every round ($r\in\{2,3,\ldots, r+1\}$) are completely determined by their respective inputs on the random tape and the content of the blackboard, 
since $\pi$ was assumed to be a deterministic protocol. 
Call the resulting protocol $\tau$.

By construction, the expected matching size of $\tau$ (over the private and public randomness $M_1,\J,\cI,G$) with respect to $H$ is 
$$  \Ex{\substack{m_1 \\ \J, \cI}} \Ex{\substack{H \sim \mu_r \\ G|\; (G(J_1)=H,m_1,\J, \cI)}} \left[ N_{\pi|m_1}(G,H)  \right] = 
\Ex{\substack{m_1 \\ \J, \cI}} \Ex{\substack{G(J_1) \sim \mu_r \\ G|\; (G(J_1),m_1,\J, \cI)}} \left[ N_{\pi|m_1}(G,G(J_1)) \right] .$$ 
By the averaging principle, there is some fixing 
of the randomness of $\tau$ that obtains (at least) the same expectation
as above with respect to $H\sim \mu_r$. Call this deterministic protocol $\tau'$.
But $\tau'$ is a deterministic $r$-round protocol over $\mu_r$, hence the inductive hypothesis asserts that 
$$\Ex{\substack{m_1 \\ \J, \cI}} \Ex{\substack{G(J_1)\sim \mu_r \\ G|\; (G(J_1),m_1,\J, \cI)}} \left[ N_{\pi|m_1}(G, G(J_1))  \right]  \;
\leq \Ex{H\sim \mu_r}[N_{\tau'}(H)]  \;
\leq t(r),$$
as claimed.

\end{proof}


We are now in shape to complete the entire proof of Theorem \ref{thm_main_lower_bound}. 
Plugging in the bounds of \eqref{eq_main_4} and Claim \ref{claim_round_elim} into equation \eqref{eq_main_2}, we have

\begin{align}
& \Ex{G\sim \mu_{r+1}} [N_{\pi}(G)]  \leq \ell\cdot n_r^2 \cdot m_r +  n_r^4 \cdot \left[ \dr^{1/2}\cdot n_r + t(r)  \right] \nonumber \\
& = \ell\cdot n_r^2 \cdot m_r +  n_r^4 \cdot \left[ \dr^{1/2}\cdot n_r + \tr  \right] \nonumber \\
& \leq n_r^4 \cdot \left[ 5n_r\cdot \dr^{1/2} +  5n_r\cdot \left(\sum_{k=0}^{r-1}\Delta_k^{1/2}\right)  \right]  
\;\;\;\;\; \text{(since $\ell\cdot n_r^2 \cdot m_r + n_r^4< 4 n^{5}_r \dr^{1/2}$)} \nonumber\\
& = 5n_r^{5} \cdot \left(\sum_{k=0}^{r}\Delta_k^{1/2}\right) = 5n_{r+1} \cdot \left(\sum_{k=0}^{r}\Delta_k^{1/2}\right)  \;\;\;\; 
\text{(since by definition, $n_{r+1}=n_r^{5}$)} \nonumber \\
& = t(r+1) - 1 \nonumber \\
& < t(r+1).
\end{align}

This proves the induction step \eqref{eq_induction_hyp}, and therefore concludes the entire proof of 
Theorem \ref{thm_main_lower_bound}. \end{proof}

\vspace{.2in}

Theorem \ref{thm_main_lower_bound} immediately implies our desired lower bound:
\begin{proof}[Proof of Theorem \ref{thm_main_lb}]
By Theorem \ref{thm_main_lower_bound}, the expected matching size of any $r$-round protocol $\pi$ under $\mu_r$ 
is at most
\begin{align*}
& \Ex{G \sim \mu_{r}} [N_{\pi}(G)] \leq t(r) = \tr \\
&= 1 + 5n_r\cdot \sum_{k=0}^{r-1}\left( \dfrac{1}{n_k}\right)^{1/2}
= 1 + 5n_r\cdot \sum_{k=0}^{r-1}\left( \dfrac{1}{\ell^{5^{k+1}}}\right)^{1/2},
\end{align*}
since by definition of $\mu_r$, $n_{r} = n_{r-1}^{5}$, and $n_0:=\ell^{5}$, hence $n_r=\ell^{5^{r+1}}$.
Since this is a doubly-exponential decaying series (and as long as $\ell$ is large enough than some absolute constant), 
 we can upper bound the sum by, say, 
\begin{align*}
&\leq 5n_r\cdot  \Delta_0^{1/5}  = 5n_r\cdot  \left(\dfrac{1}{\ell^{5}}\right)^{1/5} \\ 
& = \frac{5}{\ell}\cdot n_r = 5\cdot \frac{n_r}{n_r^{1/5^{r+1}}} = 5\cdot n_r^{1-1/5^{r+1}},
\end{align*}
since, by property (4) in Remark \ref{rem_properties}, $n_r=\ell^{5^{r+1}} \Longleftrightarrow \ell = n_r^{1/5^{r+1}}$. 

\end{proof}


\bibliographystyle{alpha}
\bibliography{refs,more-refs}

\end{document}